%% file: main.tex
\documentclass[11pt]{article}
\usepackage{amsmath,amsfonts,latexsym,amssymb,graphicx, amsthm}
\usepackage{fullpage,color}
\usepackage{url,hyperref}
\usepackage{enumerate}
\usepackage{authblk}
\usepackage[linesnumbered,boxed,ruled,vlined]{algorithm2e}
\usepackage{verbatim}
\usepackage{bbding}

\newtheorem{defn}{Definition}

\newcommand{\opt}{\Delta^*}

\newcommand{\tree}{\textsf{T}}

\newtheorem{lemma}{Lemma}

\newtheorem{theorem}[lemma]{Theorem}

\newtheorem{corollary}[lemma]{Corollary}

\begin{document}
\title{Near-linear Time Algorithm for Approximate Minimum Degree Spanning Trees \thanks{This work has been supported in part by the Zhongguancun Haihua Institute for Frontier Information Technology.}}

\author[1]{Ran Duan\thanks{duanran@mail.tsinghua.edu.cn}}
\author[1]{Haoqing He\thanks{hehq13@mails.tsinghua.edu.cn}}
\author[1]{Tianyi Zhang\thanks{tianyi-z16@mails.tsinghua.edu.cn}}

\affil[1]{Institute for Interdisciplinary Information Sciences, Tsinghua University}

\maketitle

\begin{abstract}
Given a graph $G = (V, E)$, we wish to compute a spanning tree whose maximum vertex degree, i.e. tree degree, is as small as possible. Computing the exact optimal solution is known to be NP-hard, since it generalizes the Hamiltonian path problem. For the approximation version of this problem, a $\tilde{O}(mn)$ time algorithm that computes a spanning tree of degree at most $\opt +1$ is previously known [F\"urer \& Raghavachari 1994]; here $\opt$ denotes the minimum tree degree of all the spanning trees. In this paper we give the first near-linear time approximation algorithm for this problem. Specifically speaking, we propose an $\tilde{O}(\frac{1}{\epsilon^7}m)$ time algorithm that computes a spanning tree with tree degree $(1+\epsilon)\opt + O(\frac{1}{\epsilon^2}\log n)$ for any constant $\epsilon \in (0,\frac{1}{6})$. Thus, when $\opt=\omega(\log n)$, we can achieve approximate solutions with constant approximate ratio arbitrarily close to 1 in near-linear time.
\end{abstract}

\clearpage

\input{part0}

\input{part2}

\bibliography{ref}



\end{document}

%% file: part0.tex
\section{Introduction}
Computing minimum degree spanning trees is a fundamental problem that has inspired a long time of research. Let $G = (V, E)$ be an undirected graph, and we wish to compute a spanning tree of $G$ whose tree degree, or maximum vertex degree in the tree, is the smallest. Clearly this problem is NP-hard as the Hamiltonian path problem can be reduced to it, and so we could only hope for a good approximation in polynomial time. The optimal approximation of this problem was achieved in \cite{furer1994approximating} where the authors proposed a \footnote{$\tilde{O}(\cdot)$ hides poly-logarithmic factors.}$\tilde{O}(mn)$ time algorithm that computes a spanning tree of tree degree $\leq \opt + 1$; conventionally $n = |V|, m = |E|$ and $\opt$ denotes the minimum tree degree of all the  spanning trees. For convenience, in this paper the degree of a vertex usually means its tree degree in the current spanning tree.

\subsection{Our results}
The major result of this paper is a near-linear time algorithm for computing minimum degree spanning trees in undirected graphs. To the best of our knowledge, this is the first near-linear time algorithm for this problem. Formally we propose the following statement.

\begin{theorem}\label{aug-path}
For any constant $\epsilon \in (0, \frac{1}{6})$, there is an algorithm that runs in  $O(\frac{1}{\epsilon^7}m\log^7 n)$ time which computes a spanning tree with tree degree at most $(1+\epsilon)\opt + \frac{576}{\epsilon^2}\log n$.
\end{theorem}


The core argument of Theorem \ref{aug-path} is that, starting from an arbitrary spanning tree, we repeatedly search for a sequence of distinct non-tree edges, named as \emph{augmenting sequence} (The formal definition is given in Section~\ref{degred}), to modify the current spanning tree which immediately reduces the degree of some high-degree vertex.
The idea of augmenting sequence is similar to~\cite{furer1994approximating}, that is, given a fixed degree bound $k$, an augmenting sequence w.r.t. the current spanning tree and $k$ is a sequence of vertex-disjoint non-tree edges $(w_1, z_1), (w_2, z_2), \cdots, (w_h, z_h)$ such that $w_1,w_2,\cdots,w_{h-1}$ have tree degree $k-1$ and $w_h,z_1,z_2,\cdots, z_h$ have tree degree $<k-1$. Also there is a vertex $w_0$ with tree degree $\geq k$ on the tree path between $w_1$ and $z_1$,  and $w_i$ for $1\leq i<h$ is on the tree path between $w_{i+1}$ and $z_{i+1}$ but not on the tree path between $w_{j}$ and $z_{j}$ for $j>i+1$. Then we can add theses edges $(w_1, z_1) \cdots, (w_h, z_h)$ to the spanning tree and delete the edges associated with $w_0,\cdots,w_{h-1}$ on the cycles formed, so the total degree of vertices with degree $\geq k$ will decrease by 1 but more degree-$(k-1)$ vertices may emerge. 

In our process of searching, similar to the blocking flow approach~\cite{dinitz70} for max-flow, we first construct a layering of the graph by the shortest length of augmenting sequences, then each time find a shortest augmenting sequence in the layering and do such tree modification by this augmenting sequence, thus after near-linear time the shortest length of augmenting sequences would increase. We repeat this until the length of the shortest augmenting sequence is longer than $\frac{1}{\epsilon}\log n$. 
When this happens, the number of layers also exceeds $\frac{1}{\epsilon}\log n$, so there are two adjacent layer whose ratio is at most $1+\epsilon$, then if the number of augmenting sequences we found are not too large (not too many new degree-$(k-1)$ vertices emerge), we can argue a $1+O(\epsilon)$ approximation for the optimal solution $\opt$. In the whole procedure of our algorithm, we can let $k=(1-O(\epsilon))\Delta$ for the degree $\Delta$ of the current spanning tree, and make $k$ increase by one after each iteration until in some iteration the sum of degree of all the vertices with degree $\ge k$ is not significantly decreased. See Section~\ref{sec:main-step}.

\subsection{Related work}
There is a line of works that are concerned with low-degree trees in weighted undirected graphs. In this scenario, the target low-degree that we wish to compute is constrained by two parameters: an upper bound $B$ on tree degree, an upper bound $C$ on the total weight summed over all tree edges. The problem was originally formulated in \cite{fischer1993optimizing}. Two subsequent papers \cite{konemann2000matter,konemann2003primal} proposed polynomial time algorithms that compute a tree with cost $\leq wC$ and degree $\leq \frac{w}{w-1}bB + \log_b n$, $\forall b, w> 1$. This result was substantially improved by \cite{chaudhuri2005would}; using certain augmenting path technique, their algorithm is capable of finding a tree with cost $\leq C$ and degree $B + O(\log n / \log\log n)$. Results and techniques from \cite{chaudhuri2005would} might sound similar to ours, but in undirected graphs we are actually faced with different technical difficulties. \cite{chaudhuri2005would}'s result was improved by \cite{goemans2006minimum} where for all $k$, a spanning tree of degree $\leq k+2$ and of cost at most the cost of the optimum spanning tree of maximum degree at most $k$ can be computed in polynomial time. The degree bound was later further improved from $k+2$ to the optimal $k+1$ in \cite{singh2007approximating}.

Another variant is minimum degree Steiner trees which is related to network broadcasting \cite{ravi1994rapid,ravi1993many,fraigniaud2001approximation}. For undirected graphs, authors of \cite{furer1994approximating} showed that the same approximation guarantee and running time can be achieved as with minimum degree spanning trees in undirected graphs, i.e., a solution of tree degree $\opt + 1$ and a running time of $\tilde{O}(mn)$. For the directed case, \cite{fraigniaud2001approximation} showed that directed minimum degree Steiner tree problem cannot be approximated within $(1-\epsilon)\log|D|, \forall \epsilon > 0$ unless $\mathrm{NP}\subseteq \mathrm{DTIME}(n^{\log\log n})$, where $D$ is the set of terminals.

The minimum degree tree problem can also be formulated in directed graphs. This problem was first studied in \cite{furer1990nc} where the authors proposed a polynomial time algorithm that finds a directed spanning tree of degree at most $O(\opt\log n)$. The approximation guarantee was improved to roughly $O(\opt + \log n)$ in \cite{dmdst,klein2004approximation} while the time complexity became $n^{O(\log n)}$. The problem becomes much easier when $G$ is acyclic, as shown in \cite{yao2008polynomial}, where a directed spanning tree of degree $\leq \opt + 1$ is computable in polynomial time. The approximation was greatly advanced to $\opt + 2$ in \cite{bansal2009additive} by an LP-based polynomial time algorithm, and this problem has become more-or-less closed since then.

\section{Preliminary}
Let $G=(V,E)$ be the graph we consider, and we assume $G$ is a connected graph. Logarithms are taken at base 2. During the execution of our algorithm, a spanning tree $\tree$ will be maintained. For every $u\in V$, let $\deg(u)$ be the tree degree of $u$ in $\tree$, and the degree of the spanning tree $\tree$ is defined as $\Delta = \max_{u\in V} \deg(u)$. Our algorithm will repeatedly modify $\tree$ to reduce its degree $\Delta$. Let $\opt$ denote the minimum degree of all the spanning trees. For each pair $u, v\in V$, let $\rho_{u, v}$ be the unique tree path that connects $u$ and $v$ in $\tree$. For each $1\leq k\leq n$, define $S_k = \{u\mid \deg(u) \geq k \}$ to be the set of vertices of degree at least $k$, $N_k = \{u\mid \deg(u) = k \}$ to be the set of vertices of degree exactly $k$, and $d_k = \sum_{u\in S_k} \deg(u)$ to be the sum of degrees of vertices in $S_k$.

\subsection{Boundary edge and boundary set}
Boundary edge and boundary set are important concepts to get the lower bound of $\opt$.
\begin{defn} \label{def:boundary}
For a graph $G = (V, E)$ and a sequence of disjoint vertex subsets $V_1, V_2, \cdots, V_l \subseteq V$, an edge $(u,v)\in E$ is called a boundary edge if $u\in V_i, v\in V_j$ for $1\leq i\neq j\leq l$, or $u\in V_i$ for some $i$ but $v\notin V_1\cup\cdots\cup V_l$. A vertex set $W$ is called a boundary set (with respect to $V_1, V_2, \cdots, V_l$), if for every boundary edge $(u,v)$, at least one of $u,v$ belongs to $W$.
\end{defn}

\begin{lemma}\label{undir-witness}
Let $V_1, V_2, \cdots, V_l\subseteq V$ be a sequence of disjoint vertex subsets, $W$ be a boundary set and $\opt$ be the minimum degree of all the spanning tree in $G$. Then, $\opt \geq \frac{l-1}{|W|}$.
\end{lemma}
\begin{proof}
By Definition~\ref{def:boundary}, every set $V_i$ can only be connected to other vertices by boundary edges, so for any spanning tree $T$ of $G$, there are at least $l-1$ boundary edges connecting $V_1, V_2, \cdots, V_l$ in $T$. Then for any boundary edge $(u,v)$, at least one of $u,v$ belongs to $W$. Thus by the pigeon-hole principle, there exists a $u\in W$ whose tree degree is $\geq \frac{l-1}{|W|}$.
\end{proof}

%% file: part2.tex
\section{A $(1+\epsilon)\opt + O(\frac{1}{\epsilon^2}\log n)$ Approximation}

Let $\epsilon \in (0, \frac{1}{48})$ be a fixed parameter. This algorithm starts from an arbitrary spanning tree $\tree$ and keeps modifying $\tree$ to decrease its tree degree $\Delta$. It consists of two phases: the \emph{large-step phase} and the \emph{small-step phase}. 

\begin{itemize}
	\item In the large-step phase, as long as $\Delta\geq \frac{10\log^2 n}{\epsilon^3}$, we repeatedly apply a near-linear time subroutine that, either $\Delta$ is reduced to $\leq (1-\epsilon)\cdot \Delta$ or a spanning tree $\tree$ is returned with the guarantee that $\Delta = (1+O(\epsilon))\opt$. 
	\item In the small-step phase, we need to deal with the situation where $\frac{9\log n}{\epsilon^2}\leq \Delta < \frac{10\log^2 n}{\epsilon^3}$. In this case we repeatedly run a weaker near-linear time subroutine that either $\Delta$ is reduced by $1$ or a spanning tree $\tree$ is returned with the guarantee that $\Delta \leq (1+O(\epsilon))\opt + O(\frac{\log n}{\epsilon^2})$.
\end{itemize}

Both phases rely on a degree reduction algorithm \textsf{AugSeqDegRed(k)}. The \textsf{AugSeqDegRed(k)} efficiently reduces the total degree of vertices with degree $\geq k$ by $1$ using an augmenting sequence technique.

For the rest of this section, we first propose and analyze the degree reduction algorithm \textsf{AugSeqDegRed} which underlies the core of our main algorithm. After that we specify how the large-step phase and the small-step phase work. Finally, we prove Theorem \ref{aug-path}.

\subsection{Degree reduction via augmenting sequences}\label{degred}

For a fixed threshold $k\leq \Delta$, a simple idea is that we repeatedly look for non-tree edges that connect two vertices of tree degree $\leq k-2$ from different components of $\tree\setminus S_k$ and add these edges to $\tree$, while at the same time we delete some edges incident on $S_k$ to eliminate cycles, so the tree degrees of vertices become more balanced. In this algorithm, we continue to explore possibilities of improving the tree structure using the idea of augmenting sequence as in~\cite{furer1994approximating}. For a non-tree edge $(u,v)$ that connects two different components $C_u,C_v$ of $\tree\setminus S_k$ where $deg(u) = k-1$, we try to add $(u,v)$ to $\tree$ and delete some edge incident on $S_k$ to eliminate cycles. At the same time, as $deg(u)$ increases to $k$, we keep looking for a sequence of distinct non-tree edges inside $C_u$ to add to $\tree$ and delete a sequence of tree edges to eliminate cycles.

A difficulty is that when the degrees of some original $(k-1)$-degree vertices decrease, it is hard to make the layering of the graph stable. Therefore, we define \emph{marked vertices} instead of the concept of the vertices with degree $\geq k-1$. Given a degree threshold $k\leq \Delta$, a vertex gets marked whenever its tree degree becomes $k-1$, and it stays marked even if its tree degree becomes below $k-2$ afterwards. We only re-initialize the set of marked vertices when we change $k$ in Section~\ref{sec:main-step}. Then we can define \emph{augmenting sequence} formally.
\begin{defn}[augmenting sequence]\label{aug-def}
An $h$-length augmenting sequence consists of a sequence of vertex-disjoint non-tree edges $(w_1, z_1), (w_2, z_2), \cdots, (w_h, z_h)\in E$ with the following properties.
\begin{enumerate}[(\romannumeral1)]
	\item $\exists w_0\in\rho_{w_1, z_1}\cap S_k$, and for all $0\leq i < h$,$w_i\in \rho_{w_{i+1}, z_{i+1}}\setminus(\bigcup_{j=i+2}^h\rho_{w_j, z_j})$.
	\item All $z_i$'s are unmarked ($\forall 1\leq i\leq h$); $w_i$'s are marked for all $1\leq i<h$ and $w_h$ is unmarked.
\end{enumerate}
\end{defn}
Then the tree can be modified by the augmenting sequence $(w_1, z_1), (w_2, z_2), \cdots, (w_h, z_h)$ by:
\begin{lemma}[tree modification]\label{aug-tree}
Given an augmenting sequence $(w_1, z_1), (w_2, z_2), \cdots, (w_h, z_h)\in E$, one can modify $\tree$ such that $d_k$ decreases and no vertices are added to $S_k$. Also $\Delta$ cannot increase.
\end{lemma}
\begin{proof}
We modify $\tree$ in an inductive way. For $i=h-1, h-2, \cdots, 0$, as $w_i\in \rho_{w_{i+1}, z_{i+1}}$, we can take an arbitrary tree edge $(w_i, x)\in \rho_{w_{i+1}, z_{i+1}}$, and then perform an update $\tree\leftarrow \tree\cup \{(w_{i+1}, z_{i+1})\} \setminus \{(w_i, x)\}$ which guarantees that $\tree$ is still a spanning tree. Because $w_j\notin \rho_{w_{i+1}, z_{i+1}}$ for $0\leq j\leq i-1$, tree update $\tree\leftarrow \tree\cup \{(w_{i+1}, z_{i+1})\} \setminus \{(w_i, x)\}$ does not change the connected components of $\tree\setminus \{w_j\}$, so the property that $w_j\in \rho_{w_{j+1}, z_{j+1}}\setminus(\bigcup_{l=j+2}^h\rho_{w_l, z_l}), \forall 0\leq j < i$ is preserved.

During the process, if for any $z_i$, $\deg(z_i)$ $(1\leq i\leq h)$ becomes $k-1$ during the process, mark $z_i$. By definition, $d_k$ decreases as $w_0$ loses a tree neighbour; plus, no vertices are newly added to $S_k$ because all $\deg(w_i), 1\leq i<h$ are unchanged and $\deg(w_h)\leq k-2$, $\deg(z_i)\leq k-2, \forall 1\leq i\leq h$. Also vertices in $S_k$ can only lose tree neighbors so $\Delta$ cannot increase.
\end{proof}

Now, back to the \textsf{AugSeqDegRed} algorithm. The core of this algorithm is that, if the currently shortest augmenting sequences have length $h$ ($h<1+\log_{1+\epsilon}n$), it searches for augmenting sequences of length $h$ and applies Lemma \ref{aug-tree} to decrease $d_k$. When there is no augmenting sequence of length $h$, it repeats this process for some larger $h$. Finally this algorithm terminates when $h \geq 1 + \log_{1+\epsilon} n$ and we prove a lower bound on $\opt$ based on the structure of $\tree$.

First, we introduce the \textsf{Layering} algorithm which computes an auxiliary layering of the graph that will also help tree modification later. Initially set $B_0\leftarrow S_k$. Inductively, suppose $B_0, B_1, \cdots, B_{h}, h\geq 0$ is already computed, then we compute the forest spanned by $\tree \setminus (\bigcup_{i=0}^{h}B_i)$; for each $u\in V\setminus (\bigcup_{i=0}^{h}B_i)$, let $C_u^{h}$ be the connected component of $\tree \setminus (\bigcup_{i=0}^{h}B_i)$ that contains $u$. If there exists an edge $(u, v)\in E$ such that both $u, v$ are unmarked vertices, and that $C_u^{h}\neq C_v^{h}$, then the algorithm terminates and reports that the shortest length of augmenting sequences is equal to $h+1$; otherwise, we compute $B_{h+1}$ to be the set of all marked vertices $u\in V\setminus (\bigcup_{i=0}^{h}B_i)$ such that there exists an unmarked adjacent vertex $v$ with $C_u^{h}\neq C_v^{h}$, and then continue until $h > 1+\log_{1+\epsilon} n$. Note that whenever $B_h=\emptyset$, $B_{h+1},\cdots,B_{\lceil 1+ \log_{1+\epsilon}n\rceil}$ are all empty. The pseudo code is shown in the \textsf{Layering} algorithm~\ref{layer}.

\begin{algorithm}\label{layer}
	\caption{\textsf{Layering}}
	$B_0\leftarrow S_k$, $h\leftarrow 0$\;
	\While{$h< 1+\log_{1+\epsilon} n$}{
		compute the forest $\{C_u^{h}\}$ spanned by $\tree \setminus (\bigcup_{i=0}^{h}B_i)$\;
		\If{exists unmarked $u, v$ such that $(u, v)\in E$, $C_u^h\neq C_v^h$}{
			\textbf{break}\;
		}\Else{
			compute $B_{h+1}$ to be the set of all marked vertices $u\in V\setminus (\bigcup_{i=0}^{h}B_i)$ such that there exists an unmarked adjacent vertex $v$ with $C_u^{h}\neq C_v^{h}$\;
			$h\leftarrow h+1$\;
		}
	}
	\Return $h$ and $B_0, B_1, \cdots, B_h$\;
\end{algorithm}

After we have invoked \textsf{Layering} and computed a sequence of vertex subsets $B_0, B_1, \cdots, B_h$ which naturally divides the graph into $h+2$ layers (including a layer of other vertices), every time we will find a length-$(h+1)$ augmenting sequence $(w_1, z_1), (w_2, z_2), \cdots, (w_{h+1}, z_{h+1})$ such that $w_i\in B_i$ for $1\leq i\leq h$, then  apply tree modifications of Lemma~\ref{aug-tree} by this augmenting sequence. Repeat this until there is no more length-$(h+1)$ augmenting sequences any more. The difficulty in searching for the shortest augmenting sequences is that, for a search that starts from a pair of adjacent and unmarked vertices $u, v$ satisfying $C_u^h\neq C_v^h$ and goes up the layers $B_h, B_{h-1}, \cdots, B_1, B_0$, not every route can reach the top layer $B_0$ because some previous $(h+1)$-length augmenting sequences have already blocked the road. Therefore, a depth-first search needs to be performed. To save running time, some tricks are needed: if a certain vertex has been searched before by some previous $(h+1)$-length augmenting sequences and has failed to lead a way upwards to $B_0$, then we \emph{tag} this vertex so that future depth-first searches may avoid this tagged vertex; if a certain edge has been searched before, then we \emph{tag} this edge whatsoever. The \textsf{AugDFS} algorithm may be a better illustration of this algorithm. The recursive algorithm \textsf{AugDFS} takes the layer number $i$ and an edge $(u,v)$ as input and keeps searching for edges between a vertex $w \in (u,v) \cap B_{i-1}$ and an unmarked vertex $z$. If such an edge is found, invoke \textsf{AugDFS} with the parameter $(i-1, (w,z))$ and return the result plus $(u,v)$. The pseudo code is shown in the \textsf{AugDFS} algorithm~\ref{dfs}. Later we will prove that \textsf{AugDFS(h+1,(u,v))} always returns an augmenting sequence if exists.

\begin{algorithm}\label{dfs}
	\caption{\textsf{AugDFS(i,(u,v))}}
	\If{$i=1$}{
		\Return $(u,v)$\;
	}
	\For{untagged $w \in \rho_{u,v}\cap B_{i-1}$}{
		\For{unmarked $z$ such that $(w, z)$ is untagged and $C_{z}^{i-2}\neq C_{w}^{i-2}$}{
		    $p_{i-1} \leftarrow$ \textsf{AugDFS(i-1,(w,z))}\;
			tag $(w, z)$\;
			\If{$p_{i-1} \neq$ null}{
				let $p_i$ be $p_{i-1}$ plus $(u,v)$\;
				\Return $p_i$\;
			}
		}
		tag $w$\;
	}
	\Return null\;
\end{algorithm}

The upper-level \textsf{AugSeqDegRed} algorithm repeatedly applies \textsf{Layering} followed by several rounds of \textsf{AugDFS}.
Each time \textsf{AugDFS} returns an augmenting sequence $p$, modify $\tree$ by Lemma~\ref{aug-tree} via $p$. The repeat-loop ends when $h\geq 1+\log_{1+\epsilon} n$. The pseudo code is shown in the \textsf{AugSeqDegRed} algorithm~\ref{augpath}.

\begin{algorithm}\label{augpath}
	\caption{\textsf{AugSeqDegRed(k)}}
	mark all degree $k-1$ vertices, unmark other vertices\;
	\Repeat{$h\geq 1+\log_{1+\epsilon}n$}{
		run \textsf{Layering} which computes $h$ and $B_0, B_1, \cdots, B_h$\;
		untag all vertices and edges\;
		\For{$(u, v)\in E$ such that $u, v$ are unmarked and adjacent, and that $C_u^h\neq C_v^h$}{
			$p\leftarrow$\textsf{AugDFS(h+1,(u,v))}\;
			\If{$p\neq$ null}{
			    modify $\tree$ by augmenting sequence $p$ via Lemma \ref{aug-tree}\;
			}
		}
	}
	\Return $\tree$\;
\end{algorithm}

Before proving termination of \textsf{AugSeqDegRed}, we first need to argue some properties of \textsf{Layering}. The following lemma will serve as the basis for our future proof.

\begin{lemma}[the blocking property]\label{block}
Throughout each iteration of the repeat-loop in \textsf{AugSeqDegRed}, for any $1\leq i < h$ and  any two adjacent vertices $u, v\in V\setminus (\bigcup_{j=0}^i B_j)$ such that $u$ is  unmarked and $C_u^i\neq C_v^i$, then $v\in B_{i+1}$. (Recall that $C_u^{h}$ is the connected component of $\tree \setminus (\bigcup_{i=0}^{h}B_i)$ that contains $u$.)
\end{lemma}
\begin{proof}
By rules of \textsf{Layering}, this blocking property holds right after \textsf{Layering} outputs them. This claim continuous to hold afterwards because tree modifications only merge components $C_u^i$'s and never split any $C_u^i$'s.
\end{proof}

Here is an important corollary of this Lemma \ref{block}.
\begin{corollary}\label{struct}
Throughout each iteration of the repeat-loop, for any $w\in B_i, 1\leq i\leq h$, suppose $w$ is adjacent to an unmarked $z$ such that $C_w^{i-1}\neq C_z^{i-1}$. Then $\rho_{w, z}$ only contains vertices from $V\setminus (\bigcup_{j=0}^{i-2} B_j)$.
\end{corollary}
\begin{proof}
Suppose otherwise, then there would be a vertex $x\in \rho_{w, z}\cap B_j, j\leq i-2$, then in this case $C_{w}^j \neq C_{z}^j$, and thus by Lemma \ref{block} $w\in B_{j+1}$ which is a contradiction as $j+1 < i$.
\end{proof}

Now we have the following lemmas:
\begin{lemma}\label{lem:disjoint}
If \textsf{AugDFS(h+1, (u,v))} returns a sequence of edges $(w_1, z_1),\cdots,(w_{h+1},z_{h+1})$, then $w_i\in B_i$ for $1\leq i\leq h$, and $w_{h+1}, z_1,\cdots z_{h+1}$ are unmarked, also the edges are vertex-disjoint.
\end{lemma}
\begin{proof}
    The initial $u,v$ are unmarked. From the algorithm, when calling \textsf{AugDFS(i, (u,v))}, we find a $w\in \rho_{u,v}\cap B_{i-1}$ and $z$ is unmarked , so the corresponding $w_i\in B_i$ for $1\leq i\leq h$, and $w_{h+1}, z_1,\cdots z_{h+1}$ are unmarked, also the vertices $\{w_i|1\leq i\leq h\}$ are distinct. To see that $w_{h+1}, z_1,\cdots z_{h+1}$ are distinct, we argue that in one execution of \textsf{AugDFS(i, (u,v))}, $w$ and $z$ have $C_{z}^{i-2}\neq C_{w}^{i-2}$ but $C_{z}^{i-3}= C_{w}^{i-3}$, since if $C_{z}^{i-3}\neq C_{w}^{i-3}$, $w$ would be in $B_{i-2}$ by the algorithm \textsf{Layering}. Thus $w_{h+1}, z_1,\cdots z_{h+1}$ are in distinct components in $\tree \setminus (\bigcup_{i=0}^{h}B_i)$.
\end{proof}

\begin{lemma}\label{dfs-correctness}
In the \textsf{AugSeqDegRed} algorithm, \textsf{AugDFS(h+1, (u,v))} returns either null or an augmenting sequence.
\end{lemma}
\begin{proof}
Assume a sequence of edges $(w_1, z_1),\cdots,(w_{h+1},z_{h+1})$ is returned by \textsf{AugDFS(h+1, (u,v))}. Property (\romannumeral2) in Definition \ref{aug-def} is proved by Lemma~\ref{lem:disjoint}. Now let us focus on property (\romannumeral1). We can take an arbitrary $w_0\in \rho_{w_1, z_1}\cap B_0$ since $C^0_{w_1}\neq C^0_{z_1}$ by the algorithm. Also since $w_i\in B_i, \forall 0\leq i\leq h$, by Corollary \ref{struct} we know $\rho_{w_i, z_i}$ does not contain any $w_j, 0\leq j\leq i-2$, so property (\romannumeral2) holds.
\end{proof}

The following statement concludes the \textsf{AugSeqDegRed} algorithm will terminate quickly.
\begin{lemma}\label{incr}
In the \textsf{AugSeqDegRed} algorithm, $h$ is increased by at least one during each repeat-loop, except the last one.
\end{lemma}
\begin{proof}
By the rules of \textsf{Layering}, it is easy to see that at the beginning when \textsf{Layering} outputs $B_0, B_1, \cdots, B_h$, the shortest length of augmenting sequence is equal to $h+1$. So it suffices to prove that by the end of this iteration the shortest augmenting sequence has length $>h+1$.

First we need to characterize all augmenting sequences using $B_0, B_1, \cdots, B_h$. Let the sequence $(w_1, z_1),$ $(w_2, z_2), \cdots, (w_l, z_l)$ be an arbitrary augmenting sequence and let $w_0$ be the $B_0$-vertex on $\rho_{w_1, z_1}$. We argue that $l\geq h+1$, and more importantly, if $l = h+1$, it must be that $w_i\in B_i, \forall 0\leq i\leq h$. We inductively prove that $w_i\in \bigcup_{j=0}^i B_j$ for $i = 0, 1, \cdots, l-1$. The basis is obvious as is required by property (\romannumeral1) in Definition \ref{aug-def}. Now assume $w_i\in B_r$ for some $r\leq i$. Then, from algorithm \textsf{Layering}, it would not be hard to see $w_{i+1}\in \bigcup_{j=0}^{r+1} B_j\subseteq \bigcup_{j=0}^{i+1} B_j$. Now, since components $\{C^r_u\}$ for $r\leq h-1$ are not connected by edges whose both endpoints are unmarked by Lemma~\ref{block}, so $\rho_{w_l, z_l}\cap \bigcup_{j=0}^{h-1} B_j = \emptyset$, and on the other hand $w_{l-1}\in \rho_{w_l, z_l}\cap\bigcup_{j=0}^{l-1} B_j$, so $l\geq h+1$. Plus, we can see from the induction that, when $l = h+1$ it must be that $w_i\in B_i, \forall 0\leq i\leq h$.

For any unmarked and adjacent vertices $u, v$ such that $C_u^h \neq C_v^h$, consider the instance of \textsf{AugDFS} with input $(h+1,(u,v))$. We make two claims.
\begin{enumerate}[(1)]
	\item If there is an $(h+1)$-length augmenting sequence ending with $(u, v)$, \textsf{AugDFS} would succeed in finding one.
	\item If it has returned null, then there would be no $(h+1)$-augmenting sequence ending with $(u, v)$ throughout the entire repeat-loop iteration.
\end{enumerate} 

If (1), (2) can be proved, then by the end of this repeat-loop iteration, there would be no $(h+1)$-length augmenting sequences because such augmenting sequence should end with a pair of adjacent unmarked vertices. Next we come to prove (1), (2).

\begin{enumerate}[(1)]
	\item The depth-first search of \textsf{AugDFS} exactly coincides with the conditions that $w_i\in B_i$, except that it skips all tagged vertices and edges. Now we prove that omitting tagged vertices and edges does not miss any $(h+1)$-length augmenting sequences. For an edge $(w, z)$ to be tagged, either a further recursion \textsf{AugDFS} has succeeded or failed in finding an augmenting sequences; in the former case, $C_w^{i-2}$ and $C_z^{i-2}$ has been merged, and so the condition $C_w^{i-2}\neq C_z^{i-2}$ would be violated afterwards; in the latter case, we would not need to recur on $(w, z)$ since the components w.r.t. $B_0,\cdots,B_{i-2}$ also can only merge. For a vertex $w$ to be tagged, we must have enumerated all of its untagged edges $(w, z)$ but failed to find any augmenting sequences, and therefore any future depth-first searches on $w$ would still end up in vain. 
	
	\item If \textsf{AugDFS} has once failed to find any augmenting sequences starting with $(u, v)$, then all vertices $w\in \rho_{u, v}\cap B_h$ visited by this instance of \textsf{AugDFS} should be tagged and they would be omitted by all succeeding instances of \textsf{AugDFS}. Therefore $\rho_{u, v}\cap B_h$ would stay unchanged since then. 
	Hence, if we re-run \textsf{AugDFS} with $h+1, (u, v)$, it will return null without any recursion because all vertices in $\rho_{u,v}\cap B_h$ are tagged.
\end{enumerate}
\end{proof}

Suppose \textsf{AugSeqDegRed} has terminated with $B_0, B_1, \cdots, B_{\lceil\log_{1+\epsilon}n + 1\rceil}$. We introduce the notion of a \emph{clean component}, a sequence of disjoint vertex subsets, and apply Lemma \ref{undir-witness} to get the lower bound on $\opt$.

\begin{defn}
After an instance of \textsf{AugSeqDegRed} has been executed, for any vertex $u\in V\setminus (\bigcup_{i=0}^h B_i)$, an arbitrary component $C_u^h, 0\leq h\leq \lceil\log_{1+\epsilon}n + 1\rceil$ is called \emph{clean} if all vertices in $C_u^h$ are unmarked.
\end{defn}

\begin{lemma}\label{aug-lower}
For any $0\leq h< \lceil\log_{1+\epsilon}n + 1\rceil$, suppose $\tree \setminus (\bigcup_{i=0}^h B_i)$ has $l$ clean components, then a lower bound holds that $\opt \geq \frac{l-1}{\sum_{i=0}^{h+1}|B_i|}$.
\end{lemma}
\begin{proof}
Since $h<\lceil\log_{1+\epsilon}n + 1\rceil$, $B_h$ is not the last one, so there is no edge connecting two unmarked vertices in different components of $\tree \setminus (\bigcup_{i=0}^h B_i)$. By Lemma~\ref{block}, any edge that connects a clean components of $\tree \setminus (\bigcup_{i=0}^h B_i)$ outwards must be incident on a vertex in $\bigcup_{i=0}^{h+1} B_i$, so $\bigcup_{i=0}^{h+1} B_i$ is a boundary set w.r.t. clean components. Therefore by Lemma \ref{undir-witness} we have $\opt \geq \frac{l-1}{|\bigcup_{i=0}^{h+1} B_i|} = \frac{l-1}{\sum_{i=0}^{h+1}|B_i|}$
\end{proof}


\begin{lemma}\label{running-time1}
    There is an implementation of \textsf{AugSeqDegRed} that runs in $O(\frac{1}{\epsilon^2}m\log^2 n)$ time.
\end{lemma}

\begin{proof}
We discuss some implementation details of \textsf{Layering}, \textsf{AugDFS} and \textsf{AugSeqDegRed} separately, and analyse their contributions to the total running time in a single run of \textsf{AugSeqDegRed}.

\begin{enumerate}[(1)]
	\item \textsf{Layering}.
	
	For every instance of \textsf{Layering}, computing the forest $\{C_u^{h}\}_{u\in V\setminus (\bigcup_{i=0}^{h}B_i)}$ can be done in a single pass of breath-first search which takes $O(m)$ time. Computing $B_{h+1}$, if necessary, is easily done by scanning the edge set $E$ which also takes $O(m)$ time. As the while-loop iterates for at most $1+\log_{1+\epsilon} n$ times, and due to Lemma \ref{incr} \textsf{Layering} is invoked for at most $1+\log_{1+\epsilon} n$ times, the overall contribution of \textsf{Layering} is $O(\frac{1}{\epsilon^2}m\log^2 n)$.
	
	\item \textsf{AugSeqDegRed}.
	
	Excluding the contributions of \textsf{AugDFS} and \textsf{Layering}, all \textsf{AugSeqDegRed} does is simply un-tagging all vertices and edges, scanning the edge set $(u, v)\in E$ and deciding if $C_u^h \neq C_v^h$ as well as modifying tree $\tree$. As tree components only get merged and never split, we can use the union-find data structure~\cite{Tarjan75} to support querying whether $C_u^h \neq C_v^h$ in $O(\alpha(n))$ time. Every tree modification involves insertions and deletions of $O(\frac{1}{\epsilon}\log n)$ edges, as well as merging $O(\frac{1}{\epsilon}\log n)$ pairs of some tree components $C_u^i$. Using the link-cut tree, every edge insertion and deletion takes update time $O(\log n)$, and every component-merging takes time $O(\alpha(n))$. Since every tree modification merges two components in $\tree\setminus S_k$, there can be at most $O(n)$ tree modifications throughout \textsf{AugSeqDegRed}. Therefore, the overall contribution of tree modifications is $O(\frac{1}{\epsilon}n\log^2 n)$. Hence, \textsf{AugSeqDegRed}'s exclusive contributions to the total running time would be $O(\frac{1}{\epsilon}m\alpha(n)\log n + \frac{1}{\epsilon}n\log^2 n)$.
	
	\item \textsf{AugDFS}.
	
	Now we analyze the overall time complexity induced by \textsf{AugDFS} invoked on line-5 of \textsf{AugSeqDegRed}. There are two technical issues to be resolved.
	\begin{enumerate}[(a)]
		\item How to enumerate untagged vertices $\in \rho_{u,v}\cap B_{i-1}$?
		
		For each $u_i \in B_i, \forall 0\leq i\leq h$, assign $u_i$ a weight of $i$; vertices that do not belong to any $B_i$ have weight $h+1$. By Corollary \ref{struct}, to enumerate vertices $\in \rho_{u,v}\cap B_{i-1}$, it suffices to enumerate the lightest vertices on $\rho_{u,v}$, which can be done using a link-cut tree data structure \cite{sleator1985self} built on $\tree$, each enumeration taking $O(\log n)$ amortized time. When a vertex gets tagged, we change its weight to $h+1$, and so future enumerations on $\rho_{u,v}\cap B_{i-1}$ may skip this tagged vertex.
		
		\item How to enumerate unmarked $z$ connected by an untagged edge $(w, z)$ such that $C_z^{i-2}\neq C_w^{i-2}$?
		
		Each $w$ decrementally maintains a list of all its neighbours. While we scan the list, if the next edge $(w, z)$ satisfies both conditions that $C_z^{i-2}\neq C_w^{i-2}$ and $z$ is unmarked, then the algorithm starts a new iteration and recur; either way we cross the edge $(w, z)$ off the list. In this way, every edge appears for at most once. Thus the total time of this part is $O(m\alpha(n))$; the additional $\alpha(n)$ factor comes from the union-find data structure that helps deciding if $C_z^{i-2}\neq C_w^{i-2}$.	
	\end{enumerate}
	
	Note that (a)'s running time is always dominated by (b)'s, then the overall complexity of \textsf{AugDFS} is $O(m\alpha(n) + \frac{1}{\epsilon}n\log^2 n)$.
\end{enumerate}

Summing up (1)(2)(3), the total running time is dominated by the time complexity of \textsf{Layering} which is $O(\frac{1}{\epsilon^2}m\log^2 n)$.
\end{proof}


\subsection{Large-step phase and small-step phase}\label{sec:main-step}

The large-step phase and small-step phase are described in the \textsf{ImprovedMDST} algorithm~\ref{improve}. In the large-step phase, we deal with the case $\Delta \geq \frac{10\log^2 n}{\epsilon^3}$. It works by invoking \textsf{AugSeqDegRed} with an incremental parameters $k$ from $(1-2\epsilon)\Delta + 1$ if $d_{k-1}\le 2d_k$. Within each iteration, if \textsf{AugSeqDegRed} fails to reduce $d_k$ by a factor of $(1 - \frac{\epsilon^2}{2\log n})$, then the algorithm reports a lower bound on $\opt$ and returns $\tree$ immediately. Otherwise, increase $k$ by $1$ and continue until $d_k$ becomes $0$. Since $d_{k+1}\leq d_k$, $d_k$ will become $0$ in at most $O(\log^2 n/\epsilon^2)$ iterations. Once $d_k = 0$, $\Delta$ must have decreased and repeat the while-loop. (Note that by Lemma~\ref{aug-tree}, $\Delta$ cannot increase during the whole algorithm.)

In the small-step phase, we only deal with $\frac{9\log n}{\epsilon^2}\leq \Delta < \frac{10\log^2 n}{\epsilon^3}$. Set $c = 2(\log_{1+\epsilon}n+2)$ and define a potential:
$$\phi(\tree) = \sum_{i = \Delta +1 - \log n}^\Delta c^i\cdot |N_i|$$
The small-step phase works by repeatedly selecting a degree $k$ that maximizes $c^k\cdot |N_k|$ and then run \textsf{AugSeqDegRed(k)} until $\Delta$ decreases. Similar with the large-step phase, if \textsf{AugSeqDegRed(k)} fails to reduce $d_k$ significantly, then the algorithm reports a lower bound on $\opt$ and returns $\tree$ immediately. Clearly $k$ must be larger than $\Delta - \log n$.

\begin{algorithm}\label{improve}
	\caption{\textsf{ImprovedMDST}}
	Let $\tree$ be a spanning tree of $G$ with tree degree $\Delta$\;
	\tcc{Large-step phase}
	\While{$\Delta \ge \frac{10\log^2 n}{\epsilon^3}$}{
		$k = (1-2\epsilon)\Delta + 1$\;
		\While{$d_k > 0$}{
			\If{$d_{k-1}\le 2d_k$}{
				$d\leftarrow d_{k}$\;
				run \textsf{AugSeqDegRed(k)}\;
				\If{$d_k > (1 - \frac{\epsilon^2}{2\log n})\cdot d$}{
					\Return $\tree$\;
				}
			}
			$k=k+1$\;
		}
		update the tree degree $\Delta$\;
	}
	\tcc{Small-step phase}
	\While{$\Delta\ge \frac{9\log n}{\epsilon^2}$}{
		\While{$\Delta$ has not changed}{
			pick a $k\in\arg\max_{i\in [\Delta+1 - \log n, \Delta]}\{c^i\cdot |N_i|\}$\;
			$d \leftarrow d_k$\;
			run \textsf{AugSeqDegRed($k$)}\;
			\If{$d_k > (1 - \frac{\epsilon^2}{2\log n})\cdot d$}{
				\Return $\tree$\;
			}
		}
		update the tree degree $\Delta$\;
	}
	\Return $\tree$\;
\end{algorithm}

\subsection*{Running time}
In the large-step phase, every iteration $d_k$ shrinks by a factor of $\leq (1 - \frac{\epsilon^2}{2\log n})$, so $d_k$ will become zero in $O(\log^2 n/\epsilon^2)$ iterations. We have:
\begin{lemma}\label{runtime1}
	The running time of the large-step phase is bounded by $O(\frac{1}{\epsilon^5}m\log^5 n)$.
\end{lemma}
\begin{proof}
	From the previous subsection we already know that \textsf{AugSeqDegRed} runs in $O(\frac{1}{\epsilon^2}m\log^2 n)$ time, so here we only need to upper bound the total number of times \textsf{AugSeqDegRed} gets invoked before $\Delta < \frac{10\log^2 n}{\epsilon^3}$ or a spanning tree $\tree$ is returned within a while-loop. Next we only focus on the previous cases because it takes a longer running time. In this case, at the end of each iteration, $d_k \leq (1-\frac{\epsilon^2}{2\log n})\cdot d$. The inside while-loop would break when $k > (1-2\epsilon)\Delta + \frac{2\log^2 n}{\epsilon^2}$ because by the time $$d_k \leq \left(1-\frac{\epsilon^2}{2\log n}\right)^{\frac{2\log^2 n}{\epsilon^2}}\cdot d_{(1-2\epsilon)\Delta} \leq \frac{d_{(1-2\epsilon)\Delta}}{n} < 1$$ 
	As $(1-2\epsilon)\Delta + \frac{2\log^2 n}{\epsilon^2} \leq (1-\epsilon)\Delta$ when $\Delta \geq \frac{10\log^2 n}{\epsilon^3}$, which means $\Delta$ has been reduced by a factor of at most $1-\epsilon$ in the end of each while-loop and there are at most $O(\frac{1}{\epsilon}\log n)$ while-loops within the large-step phase. In summary, the total running time of the large-step phase is $O\left(\frac{\log^2 n}{\epsilon^2}m \times \frac{\log^2 n}{\epsilon^2} \times \frac{\log n}{\epsilon}\right) = O\left(m\cdot\frac{\log^5 n}{\epsilon^5}\right)$.
\end{proof}

In the small-step phase, every iteration $\phi(\tree)$ shrinks by a factor of $\le 1 - \frac{\epsilon^2}{5\log^2 n}$, after $O(\frac{\log^3 n}{\epsilon^2})$ rounds, $\phi(\tree)$ will be smaller than $c^{\Delta}$.

\begin{lemma}\label{runtime2}
	The running time of the small-step phase is bounded by $O(\frac{1}{\epsilon^7}m\log^7 n)$.
\end{lemma}

\begin{proof}
	We already know that \textsf{AugSeqDegRed} runs in $O(\frac{1}{\epsilon^2}m\log^2 n)$ time. Now we study how many rounds of \textsf{AugSeqDegRed} could be invoked before $\Delta$ changes or this algorithm returns $\tree$ within a while-loop. We only focus on the previous cases because it takes a longer running time. For one execution of \textsf{AugSeqDegRed}, let $N^\prime_k, d^\prime_k, \tree^\prime$ $(k\in [\Delta + 1-\log n, \Delta])$ be snapshots of $N_k,d_k,\tree$ right before we execute \textsf{AugSeqDegRed}, and here we consider the case when $\Delta$ is not changed and $d_k \leq (1-\frac{\epsilon^2}{2\log n})\cdot d^\prime_k$.

	Next we analyse how $\phi(\tree)$ has decreased. The potential before the change is $\phi(\tree^\prime) = \sum_{i=\Delta+1-\log n}^\Delta c^i\cdot |N_k^\prime|$. Every time \textsf{AugSeqDegRed} modified $\tree$, at least one vertex in $S_k$ lost a tree edge and at most $2+\log_{1+\epsilon} n$ vertices with degree $<k-1$ gained a tree edge, and then the total loss of $\phi(\tree)$ would be at least
$$\begin{aligned}
	&(c^k - c^{k-1}) - (2+\log_{1+\epsilon}n)\cdot (c^{k-1} - c^{k-2})\geq (c^{k-1} - c^{k-2})(c - 2-\log_{1+\epsilon}n)\\
	&= c^k\cdot \left(1 - \frac{1}{c}\right)\cdot \left(1 - \frac{2+\log_{1+\epsilon} n}{c}\right) \geq c^k\cdot \left(1 - \frac{1}{c}\right)\cdot \frac{1}{2} > 0.4\cdot c^k
\end{aligned}$$

	After executing \textsf{AugSeqDegRed}, $d_k$ has decreased by $d_k^\prime - d_k\geq \frac{\epsilon^2}{2\log n}d_k^\prime \geq \frac{\epsilon^2}{2\log n}\cdot k|N_k^\prime|$. For each tree modification via Lemma~\ref{aug-tree}, at most two vertices in $S_k$ lost one degree (only removing the edge connected to $w_0$ affects $S_k$), which makes $d_k$ decreased by at most $2k$. So there are at least $(d_k^\prime - d_k) / (2k) \geq \frac{\epsilon^2}{4\log n}\cdot |N_k^\prime|$ tree modifications via Lemma~\ref{aug-tree} to $\tree$. Therefore, $$\phi(\tree)\leq \phi(\tree^\prime) - (0.4\cdot c^k)\cdot \left(\frac{\epsilon^2}{4\log n}\cdot |N_k^\prime|\right)\leq \left(1-\frac{0.1\epsilon^2}{\log^2 n}\right)\phi(\tree^\prime)$$ The second inequality holds by maximality of $c^k\cdot |N_k^\prime|$ which implies $c^k\cdot |N_k^\prime|\geq \frac{1}{\log n}\cdot \phi(\tree^\prime)$.

	In a nutshell, $\phi(\tree)$ has decreased by a factor of at most $1 - \frac{0.1\epsilon^2}{\log^2 n}$. As long as $\Delta$ has not changed, $\phi(\tree)$ belongs to the interval $(c^{\Delta}, n\cdot c^\Delta)$, and consequently, $\phi(\tree)$ could suffer at most $-\log_{1 - \frac{0.1\epsilon^2}{\log^2 n}}n = O(\frac{\log^3 n}{\epsilon^2})$ rounds of \textsf{AugSeqDegRed} before $\Delta$ decreases. There are at most $O(\frac{\log^2 n}{\epsilon^3})$ while-loops in the small-step phase because each while-loop reduces $\Delta$ by at least $1$.

	In summary, the total running time of the small-step phase is $O\left(\frac{\log^2 n}{\epsilon^2}m \times \frac{\log^3 n}{\epsilon^2} \times \frac{\log^2 n}{\epsilon^3}\right) = O\left(m\cdot\frac{\log^7 n}{\epsilon^7}\right)$
\end{proof}

\subsection*{Approximation guarantee}

When a spanning tree $\tree$ is returned within the large-step phase or the small-step phase, the vertex subsets $B_0,B_1,\cdots,B_{\lceil 1+\log_{1+\epsilon} n\rceil}$ created by \textsf{AugSeqDegRed} satisfies the blocking property (see Lemma~\ref{block}). By Lemma~\ref{aug-lower}, there is a lower bound on $\opt$ for each vertex set $B_h, 0\le h< \lceil 1+\log_{1+\epsilon} n\rceil$ as long as we get the lower bound on the number of clean components in $\tree \setminus (\bigcup_{i=0}^h B_i)$. The following two statements show the lower bound on $\opt$.

\begin{lemma}\label{num-comp}
For any vertex subset $B$ and any spanning tree $\tree$, the number of connected components in $\tree\setminus B$ is at least $\sum_{u\in B}\deg(u) - 2|B| + 2$.
\end{lemma}
\begin{proof}
Note that there are at least $\sum_{u\in B}\deg(u) - |B| + 1$ tree edges incident on $B$, and so removing all of these edges would break $\tree$ into $\geq \sum_{u\in B}\deg(u) - |B| + 2$ components. Therefore, excluding singleton components formed by $B$, there are $\geq\sum_{u\in B}\deg(u) - 2|B| + 2$ components from $\tree\setminus B$. 
\end{proof}

\begin{lemma}\label{num-clean}
If a spanning tree $\tree$ is returned within the large-step phase or the small-step phase and $k$ is the parameter of the last invoked \textsf{AugSeqDegRed}, for any $0\le h < \lceil 1+\log_{1+\epsilon} n\rceil$, the number of clean components in $\tree \setminus (\bigcup_{i=0}^h B_i)$ is at least $k\cdot(1 - 4\epsilon)\sum_{i=0}^h |B_i| + 1$ for $\epsilon\in (0, \frac{1}{48})$. Furthermore,
$$\opt \geq k(1-4\epsilon)\cdot \frac{\sum_{i=0}^h |B_i|}{\sum_{i=0}^{h+1} |B_i|}$$
\end{lemma}
\begin{proof}
	By Lemma \ref{num-comp}, the number of tree components in $\tree \setminus (\bigcup_{i=0}^h B_i)$ is at least $$\sum_{u\in \bigcup_{i=0}^h B_i}\deg(u) - 2\left|\bigcup_{i=0}^h B_i\right| + 2$$
	Let $d_k^\prime, d_{k-1}^\prime$, $S_{k-1}^\prime$ and $S_k^\prime$ be snapshots of $d_k, d_{k-1}$, $S_{k-1}$ and $S_k$ right before the last instance of \textsf{AugSeqDegRed} started and let $M$ be the set of all marked vertices $\notin S_{k-1}^\prime$ (i.e., vertices that are initially unmarked) by the end of \textsf{AugSeqDegRed}. Then, the number of clean components in $\tree \setminus (\bigcup_{i=0}^h B_i)$ is at least
	$$\sum_{u\in \bigcup_{i=0}^h B_i}\deg(u) - 2\sum_{i=0}^h |B_i| + 2 - |M\cup S_{k-1}^\prime|$$
	
	The argument consists of a lower bound on $\sum_{u\in \bigcup_{i=0}^h B_i}\deg(u)$ and an upper bound on $|M\cup S_{k-1}^\prime|$. 
	
	\begin{enumerate}[(1)]
		\item Lower bound on $\sum_{u\in \bigcup_{i=0}^h B_i}\deg(u)$.
		
		By the \textsf{Layering} algorithm $B_0 = S_k$, then we have $\sum_{u\in B_0}\deg(u) = d_k$.
	
		For any vertex $u\in \bigcup_{i=1}^h B_i$, $deg(u) = k-1$ by the time $u$ was first added to some $B_i$. After that, $\deg(u)$ could only decrease when we modify $\tree$ by an augmenting sequence $(w_1,z_1),\cdots,(w_t,z_t)$ where $u=w_j$ for some $1\le j\le t$. Since $t \le \lceil 1+\log_{1+\epsilon} n\rceil$, during a tree modification, at least one vertex in $S_k$ loses one degree and at most $\lceil 1+\log_{1+\epsilon} n\rceil$ vertices in $\bigcup_{i=1}^h B_i$ lose one degree separately. As the total number of the degree loss in $S_k$ is $(d_k^\prime - d_k)$, we have
	
		$$\sum_{u\in \bigcup_{i=1}^h B_i}\deg(u)\geq (k-1)\sum_{i=1}^h |B_i| - (d_k^\prime - d_k)\lceil 1 + \log_{1+\epsilon} n \rceil$$
	
		Since $d_k > (1 - \frac{\epsilon^2}{2\log n})\cdot d_k^\prime$, we get a lower bound on $\sum_{u\in \bigcup_{i=0}^h B_i}\deg(u)$,
	
		$$\begin{aligned}
		\sum_{u\in \bigcup_{i=0}^h B_i}\deg(u) &\geq d_k + (k-1)\sum_{i=1}^h |B_i| - (d_k^\prime - d_k)\lceil 1 + \log_{1+\epsilon} n \rceil\\
		&\geq (k-1)\sum_{i=1}^h |B_i| + \left(1 - \frac{\epsilon^2}{2\log n}\right)d_k^\prime - \frac{\epsilon^2}{2\log n}(2+ \log_{1+\epsilon} n) d_k^\prime\\
		&\geq (k-1)\sum_{i=1}^h |B_i| + \left(1 - \frac{3\epsilon^2}{2\log n}- \frac{\epsilon}{2}\right)d_k^\prime
		\end{aligned}$$
		
		\item Upper bound on $|M|$.
		
		The argument is similar to (1). An unmarked vertex $u$ is marked only when we modify $\tree$ by an augmenting sequence $(w_1,z_1),\cdots,(w_t,z_t)$ where $u=z_j$ for some $1\le j\le t$ or $u=w_t$. Since $t \le 1+\log_{1+\epsilon} n$, during a tree modification, at least one vertex in $S_k$ loses one degree and at most $2+\log_{1+\epsilon} n$ unmarked vertices are marked. Then we get a upper bound on $|M|$.
		
		$$|M| \leq (d_k^\prime - d_k)(2 + \log_{1+\epsilon} n) \leq \epsilon \cdot d_k^\prime$$
		
		\item Upper bound on $|S_{k-1}^\prime|$.

		First we claim $\frac{d_k^\prime}{d_{k-1}^\prime} \geq \frac{1}{\epsilon(k-1)}$ when $k\geq \frac{9\log n}{\epsilon^2} - \log n \geq \frac{8\log n}{\epsilon^2}$ and $n > 2^\epsilon$. In the large-step phase, the inequality holds since $d_k^\prime \geq \frac{1}{2} d_{k-1}^\prime$. In the small-step phase, by maximality of $c^k\cdot |N_k|$, we have $|N_k|\geq \frac{1}{c}\cdot |N_{k-1}|$. Then, $$\frac{d_k^\prime}{d_{k-1}^\prime} = \frac{\sum_{i=k}^\Delta i|N_i|}{\sum_{i=k-1}^\Delta i|N_i|} > \frac{k|N_k|}{k|N_k| + (k-1)|N_{k-1}|}\geq \frac{1}{1 + \frac{c(k-1)}{k}} > \frac{1}{c+1} \geq \frac{1}{\epsilon(k-1)}$$
		The last inequality holds by $c = 2\log_{1+\epsilon} n + 4 \leq \frac{2\ln 2}{\epsilon}\log n + 4$, $k\cdot\epsilon\geq \frac{8\log n}{\epsilon}$ and $\epsilon<\log n\leq \frac{\log n}{\epsilon}$. Then we have the upper bound on $|S_{k-1}^\prime|$:
		$$|S_{k-1}^\prime| \leq \frac{d_{k-1}^\prime}{k-1} \leq \epsilon\cdot d_k^\prime$$
	\end{enumerate}
	
	Summing up (1), (2), (3), for $n > 2$ and $\epsilon\in (0,\frac{1}{48})$, $k\geq \frac{9\log n}{\epsilon^2} - \log n$, we have
	$$\begin{aligned}
	&\sum_{u\in \bigcup_{i=0}^h B_i}\deg(u) - 2\sum_{i=0}^h |B_i| + 2 - |M\cup S_{k-1}^\prime|\\
	&\geq \left(1-\frac{3\epsilon^2}{2\log n} -2.5\epsilon\right)\cdot d_k^\prime  + (k-1)\sum_{i=1}^h |B_i| - 2\sum_{i=0}^h |B_i| + 2\\
	&\geq \left(1-2.6\epsilon\right)\cdot k|B_0|  - 2|B_0| + (k-3)\sum_{i=1}^h |B_i| + 2\\
	&\geq k(1-4\epsilon)\cdot \sum_{i=0}^h |B_i| +2
	\end{aligned}$$
	We apply Lemma~\ref{aug-lower}, and conclude the proof
	$$\opt \geq \frac{k(1-4\epsilon)\cdot\sum_{i=0}^h |B_i| + 1}{\sum_{i=0}^{h+1} |B_i|} \geq k(1-4\epsilon)\cdot \frac{\sum_{i=0}^h |B_i|}{\sum_{i=0}^{h+1} |B_i|}$$
\end{proof}

In the following two statements, we combine all the inequalities for each $B_h$ and get the upper bound on $\Delta$ with $\opt$ in both the large-step phase and the small-step phase.

\begin{lemma}\label{guarantee1}
	When a spanning tree $\tree$ is returned within the large-step phase, it must be that $\Delta \leq (1+8\epsilon)\cdot \opt$ for $\epsilon\in (0, \frac{1}{48})$.
\end{lemma}
\begin{proof}
	Consider the most recent execution of \textsf{AugSeqDegRed} before returning. By the previous subsection, this instance of \textsf{AugSeqDegRed} has created a sequence of disjoint vertex subsets $B_0, B_1, \cdots, B_{1+\log_{1+\epsilon} n}$ that satisfy the blocking property. By the pigeon-hole principle, there exists an $h$ such that $\frac{\sum_{i=0}^h |B_i|}{\sum_{i=0}^{h+1} |B_i|} \geq \frac{1}{1+\epsilon}$. Then by Lemma \ref{num-clean}, (recall that in the large-step phase $k>(1-2\epsilon)\Delta$)
	$$\opt\geq k(1-4\epsilon)\cdot \frac{1}{1+\epsilon} > \frac{1-6\epsilon + 8\epsilon^2}{1+\epsilon}\Delta$$
	or equivalently, $\Delta \leq \frac{1+\epsilon}{1-6\epsilon + 8\epsilon^2}\opt < (1+8\epsilon)\opt$ when $\epsilon \in (0, \frac{1}{48})$.
\end{proof}

\begin{lemma}\label{guarantee2}
	When a spanning tree $\tree$ is returned within the small-step phase, it must be that $\Delta\leq (1+6\epsilon)\opt + \log n$ for $\epsilon\in (0, \frac{1}{48})$.
\end{lemma}
\begin{proof}
	Consider the most recent execution of \textsf{AugSeqDegRed} before returning. By the previous subsection, this instance of \textsf{AugSeqDegRed} has created a sequence of disjoint vertex subsets $B_0, B_1, \cdots, B_{1+\log_{1+\epsilon} n}$ that satisfies the blocking property. By the pigeon-hole principle, there exists an $h$ such that $\frac{\sum_{i=0}^h |B_i|}{\sum_{i=0}^{h+1} |B_i|} \geq \frac{1}{1+\epsilon}$. Then by Lemma \ref{num-clean}, (recall that in the small-step phase, $k>\Delta-\log n)$)
	 $$\opt\geq k(1-4\epsilon)\cdot \frac{1}{1+\epsilon} > \frac{1-4\epsilon}{1+\epsilon}(\Delta - \log n)$$ or equivalently, $\Delta \leq \frac{1+\epsilon}{1-4\epsilon}\opt + \log n < (1+6\epsilon)\opt + \log n$ for $\epsilon\in (0, \frac{1}{48})$.
\end{proof}

Now we can finish the proof of Theorem~\ref{aug-path}

\begin{proof}[Proof of Theorem~\ref{aug-path}]
	We claim that, for any constant $\epsilon \in (0, \frac{1}{6})$, the \textsf{ImprovedMDST} algorithm computes a spanning tree with tree degree at most $(1+\epsilon)\opt + \frac{576}{\epsilon^2}\log n$ in $O(\frac{1}{\epsilon^7}m\log^7 n)$ time (by resetting $\epsilon\rightarrow 8\epsilon'$ where $\epsilon'$ is the $\epsilon$ in previous analysis).

	By Lemma~\ref{runtime1} and Lemma~\ref{runtime2}, we know the total running time of the \textsf{ImprovedMDST} algorithm is bounded by $O(\frac{1}{\epsilon^7}m\log^7 n)$. The degree analysis is divided into three cases:
	\begin{itemize}
		\item If $\tree$ is returned after the small-step phase, then $\Delta < \frac{576}{\epsilon^2}\log n$ for $\epsilon\in (0, \frac{1}{6})$.

		\item If $\tree$ is returned during the large-step phase, by Lemma~\ref{guarantee1}, $\Delta \leq (1 + \epsilon)\cdot \opt$ for $\epsilon\in (0, \frac{1}{6})$.

		\item If $\tree$ is returned during the small-step phase, by Lemma~\ref{guarantee2}, $\Delta\leq (1+ \frac{3\epsilon}{4})\opt + \log n$ for $\epsilon\in (0, \frac{1}{6})$.
	\end{itemize}

	In summary, the \textsf{ImprovedMDST} algorithm computes a spanning tree with tree degree $(1+\epsilon)\opt + \frac{576}{\epsilon^2}\log n$ in $O(\frac{1}{\epsilon^7}m\log^7 n)$ time.
\end{proof}

%% file: main.bbl
\begin{thebibliography}{10}

\bibitem{bansal2009additive}
Nikhil Bansal, Rohit Khandekar, and Viswanath Nagarajan.
\newblock Additive guarantees for degree-bounded directed network design.
\newblock {\em SIAM Journal on Computing}, 39(4):1413--1431, 2009.

\bibitem{chaudhuri2005would}
Kamalika Chaudhuri, Satish Rao, Samantha Riesenfeld, and Kunal Talwar.
\newblock What would \text{Edmonds} do? \text{Augmenting} paths and witnesses
  for degree-bounded \text{MSTs}.
\newblock {\em Lecture notes in computer science}, 3624:26, 2005.

\bibitem{dinitz70}
Yefim Dinitz.
\newblock Algorithm for solution of a problem of maximum flow in networks with
  power estimation.
\newblock {\em Soviet Math. Dokl.}, 11:1277--1280, 01 1970.

\bibitem{fischer1993optimizing}
Ted Fischer.
\newblock Optimizing the degree of minimum weight spanning trees.
\newblock Technical report, Cornell University, 1993.

\bibitem{fraigniaud2001approximation}
Pierre Fraigniaud.
\newblock Approximation algorithms for minimum-time broadcast under the
  vertex-disjoint paths mode.
\newblock {\em Algorithms—ESA 2001}, pages 440--451, 2001.

\bibitem{furer1990nc}
Martin F{\"u}rer and Balaji Raghavachari.
\newblock An \text{NC} approximation algorithm for the minimum degree spanning
  tree problem.
\newblock In {\em Proc. of the 28th Annual Allerton Conf. on Communication,
  Control and Computing}, pages 274--281, 1990.

\bibitem{furer1994approximating}
Martin F{\"u}rer and Balaji Raghavachari.
\newblock Approximating the minimum-degree \text{Steiner} tree to within one of
  optimal.
\newblock {\em Journal of Algorithms}, 17(3):409--423, 1994.

\bibitem{goemans2006minimum}
Michel~X Goemans.
\newblock Minimum bounded degree spanning trees.
\newblock In {\em Foundations of Computer Science, 2006. FOCS'06. 47th Annual
  IEEE Symposium on}, pages 273--282. IEEE, 2006.

\bibitem{klein2004approximation}
Philip~N Klein, Radha Krishnan, Balaji Raghavachari, and R~Ravi.
\newblock Approximation algorithms for finding low-degree subgraphs.
\newblock {\em Networks}, 44(3):203--215, 2004.

\bibitem{konemann2000matter}
Jochen K{\"o}nemann and R~Ravi.
\newblock A matter of degree: Improved approximation algorithms for
  degree-bounded minimum spanning trees.
\newblock In {\em Proceedings of the thirty-second annual ACM symposium on
  Theory of computing}, pages 537--546. ACM, 2000.

\bibitem{konemann2003primal}
Jochen K{\"o}nemann and R~Ravi.
\newblock Primal-dual meets local search: approximating mst's with nonuniform
  degree bounds.
\newblock In {\em Proceedings of the thirty-fifth annual ACM symposium on
  Theory of computing}, pages 389--395. ACM, 2003.

\bibitem{dmdst}
Radha Krishnan and Balaji Raghavachari.
\newblock The directed minimum-degree spanning tree problem.
\newblock In {\em FSTTCS}, volume 2245, pages 232--243. Springer, 2001.

\bibitem{ravi1994rapid}
R~Ravi.
\newblock Rapid rumor ramification: Approximating the minimum broadcast time.
\newblock In {\em Foundations of Computer Science, 1994 Proceedings., 35th
  Annual Symposium on}, pages 202--213. IEEE, 1994.

\bibitem{ravi1993many}
R~Ravi, Madhav~V Marathe, SS~Ravi, Daniel~J Rosenkrantz, and Harry~B Hunt~III.
\newblock Many birds with one stone: Multi-objective approximation algorithms.
\newblock In {\em Proceedings of the twenty-fifth annual ACM symposium on
  Theory of computing}, pages 438--447. ACM, 1993.

\bibitem{singh2007approximating}
Mohit Singh and Lap~Chi Lau.
\newblock Approximating minimum bounded degree spanning trees to within one of
  optimal.
\newblock In {\em Proceedings of the thirty-ninth annual ACM symposium on
  Theory of computing}, pages 661--670. ACM, 2007.

\bibitem{sleator1985self}
Daniel~Dominic Sleator and Robert~Endre Tarjan.
\newblock Self-adjusting binary search trees.
\newblock {\em Journal of the ACM (JACM)}, 32(3):652--686, 1985.

\bibitem{Tarjan75}
Robert~Endre Tarjan.
\newblock Efficiency of a good but not linear set union algorithm.
\newblock {\em J. ACM}, 22(2):215--225, April 1975.

\bibitem{yao2008polynomial}
Guohui Yao, Daming Zhu, Hengwu Li, and Shaohan Ma.
\newblock A polynomial algorithm to compute the minimum degree spanning trees
  of directed acyclic graphs with applications to the broadcast problem.
\newblock {\em Discrete Mathematics}, 308(17):3951--3959, 2008.

\end{thebibliography}
